\pdfoutput=1

\documentclass{fundam}
\usepackage{hyperref}
\usepackage[normalem]{ulem}
\usepackage{multirow}
\usepackage{booktabs}
\usepackage{algorithm}
\usepackage{algorithmic}
\usepackage{float}
\usepackage{mathtools}
\usepackage{flushend}
\usepackage{ulem}
\usepackage{amssymb}
\usepackage{hyperref}
\usepackage{multirow}
\usepackage{txfonts}
\usepackage{todonotes}
\usepackage{algorithmic}
\usepackage{subfloat,subfig}

\newcommand{\be}{\begin{enumerate}}
\newcommand{\ee}{\end{enumerate}}
\def\ignore#1{}

\def\intersection{\cap}

\def\dom{\ensuremath{\mathrm{dom}}}

\def\val{\ensuremath{\mathrm{val}}}
\newcommand{\nullname}{\texttt{null}}
\newcommand{\sch}{\mathop{\mathrm{sch}}}

\newcommand{\BigO}[1]{\ensuremath{\operatorname{O}\bigl(#1\bigr)}}

\ignore{

\newcommand{\danielinline}[1]{\todo[inline,backgroundcolor=blue!20!white]{#1}}

\newcommand{\antonioinline}[1]{\todo[inline,backgroundcolor=green!20!yellow]{#1}}

}

\newtheorem{defin}{Definition}[section]
\newcommand{\union}{\cup}

\def\fd#1#2{\mathit{#1} \to \mathit{#2}}

\begin{document}
\setcounter{page}{1}
\issue{to appear}

\title{On Desirable Semantics of Functional Dependencies over
Databases with Incomplete Information}
\address{A. Badia, Duthie Hall 227, University of Louisville,
Louisville KY 40292, USA}


\author{Antonio Badia\\
CECS Department\\
University of Louisville\\
Louisville KY 40292, USA
\and Daniel Lemire\\
LICEF\\
Universit\'e du Qu\'ebec\\
5800 Saint-Denis, Montreal, QC, H2S 3L5 Canada}
\maketitle

\runninghead{A. Badia and D. Lemire}{Functional Dependencies over
Databases with Incomplete Information}

\begin{abstract}
Codd's relational model  describes just one possible world. To better
 cope with incomplete  information,  extended database models 
 allow several possible worlds. Vague tables are one such convenient
 extended model where attributes accept sets of possible values
 (e.g., the manager is either Jill or Bob). 
 However, conceptual database design in such cases
  remains an open problem. 
In particular,  there is no canonical definition of functional
 dependencies (FDs) over possible worlds (e.g., each employee has just
 one manager). We identify  several desirable properties that the
 semantics of such FDs  should  meet including Armstrong's axioms, the
 independence from irrelevant attributes, seamless satisfaction and
 implied by strong satisfaction.  We show that we can define FDs such
 that they have all our desirable properties over vague tables. 
However, we also show that no notion of FD can satisfy all our
 desirable properties over a more general model (disjunctive
 tables). Our work formalizes a trade-off between having a general
 model and having well-behaved FDs.
\end{abstract}


\section{Introduction}

We often still teach database design like we did decades ago. Yet many
of the underlying assumptions made by database textbooks are fading.  
For example, it is no longer always possible to assume that
information systems are strongly consistent when even a modest organization
might have hundreds of different databases   hosted all over the world.
Our information systems have to reflect different viewpoints and accommodate
uncertainties and disagreements.

Nevertheless, we must maintain basic relationships between attributes
if we hope to use the data. One such fundamental relationship is given by 
functional dependencies (hereafter, FDs). For example, we may know
that all employees have only one superior (written
``$\textrm{employee} \to \textrm{superior}$''). Mathematically, we 
require that the relation between employee and superior be
a {function}: each possible employee must have one and only one
superior. 

Identifying such relationships is part of conceptual design: determining what belongs in the database and
how it should be structured. Conceptual design under the assumption
that we have complete information might be considered a solved
problem. 


Unfortunately, our knowledge is often
incomplete~\cite{Badia:2011:CAR:2070736.2070750}---and that is even
more apparent today with our highly distributed and heterogeneous
systems. When merging two 
databases, we might find that one assigns Jill as John's superior,
whereas the other assigns Bob as John's superior. Such a discrepancy
may come, for instance, when an update was 
propagated to one database but not yet to the other. We can reasonably
assume that John's superior is either Jill or Bob but it may  be
impossible to  clear up the ambiguity quickly. 
These problems are frequent and costly: 
Brodie and Liu~\cite{brodie2010} estimated that 40\% of the cost
associated with information systems is due to data integration
problems. 

What should we do in such cases? One possibility is to declare John's
superior to be unknown. We might represent an unknown value as
a \nullname{} marker in a traditional 
relational database~\cite{Codd:1986:MIR:16301.16303}. Thankfully, FDs
can be enforced despite \nullname{} markers~\cite{Badia15052014}.

However, using \nullname{} markers is not always satisfactory. In our
example with John and his superior, we expect the superior to be
either Jill or Bob. We want the database to retain this information
while enforcing FDs.  That is, we want to allow sets of possible
values (John's superior is either Jill or Bob): we call the result
a \emph{vague} table.


 Such a proposal is not new: several models of incomplete information
 beside vague tables have been proposed, each with different
 expressiveness and  properties (see Das Sarma et
 al.~\cite{sarmauncertain:2009} for an  overview\footnote{In Das Sarma et
al.~\cite{sarmauncertain:2009} and other works,  vague tables are
called {attribute-or} tables.})  Yet there is no
 agreed-upon concept of FDs for databases with incomplete
 information. This leaves practitioners and textbooks with little
 guidance. But without such guidance,
 practical implementations are often ad hoc.

We approach the problem of defining a suitable notion
of FD over incomplete information from a novel perspective. As
there are already a variety of proposed definitions over a diversity
of  models of incomplete information, we first attempt to come up with
a list of  desirable properties that the definition of a practical
concept of FD might have, regardless of the type of model involved. We
justify each property on intuitive grounds. 
 We then review several past proposals and show that they do
not satisfy some of these desirable properties. This leaves open the
question of whether a notion of FD can exist that has all the desired
properties over some model of incomplete information. We show that
there cannot be such a notion over one such model (disjunctive
tables). We then propose a definition of FD over vague
tables, P-functional dependency or PFD, that has a simple
interpretation and  provides all our desirable properties
over vague tables. In this sense, PFDs serve as an example to show
that our desirable properties can be satisfied over vague tables.

We emphasize that the issue under discussion here is the intended
interpretation of $X \to Y$ under different models. Thus, in the
following, references to a 'concept' or 'notion' of FD refer to the
{\em semantics of FDs}; we explore the characteristics that a
'desirable' semantics should have in order to develop a viewpoint
that allow us to analyze and compare existing definitions as well as
our own. 

The fact that a concept of FD over disjunctive tables cannot have all
the desired properties is in agreement with the observation
in Das Sarma et al.~\cite{sarmauncertain:2009} that more expressive
models tend to be more complex and more difficult to reason about. Our
results show that it may also be difficult to come up with a notion of
FD that achieves general intuitive agreement over disjunctive tables or
more sophisticated models. Hence, this motivates the study and
application of more restrictive models like vague tables, in spite of
their expressive shortcomings.

\section{Related Work}
\label{section:relatedwork}

There is a  large body of work on  FDs for
extended models. We can 
distinguish three main approaches: 
\begin{enumerate}
\item work that deals with incomplete
data by using \nullname{}s or sets of possible values (including
disjunctive databases)~\cite{levene1998axiomatisation,link2014relational};   
\item work that adds  information other than values
(such as possibilistic/probabilistic databases)~\cite{bosc2003impact,bosc2009functional,bosc2013,wu2010maintaining,lu2009maintaining};  
\item and work that does  not deal with simple values (fuzzy databases, where
values can be fuzzy functions)~\cite{doi:10.1080/03081079.2013.798909,qureshi2012using,cordero2011efficient,liu2012handling,bosc1998fuzzy}. 
\end{enumerate}

In the first approach, it is common to consider the
database as denoting a set of possible worlds. 
This line of 
research has applied the tools of modal logic to the study of data
dependencies~\cite{hartmann2010data}. 


The second approach, based on possibilistic and probabilistic
databases, can be considered an extension of the the possible
worlds
framework~\cite{bosc2003impact,bosc2009functional,bosc2013}. However,
Link and Prade~\cite{link2014relational} assign
possibilities to tuples, not to values. Based on these possibilistic tuples,
    possible worlds are generated as a nested chain: the smallest
    world is the one with only fully possible tuples; the largest
    world contains all tuples. Certainty degrees are attached to standard
    FDs, based on the possibility degree of the smallest world
    in which they are violated. 

Work on fuzzy databases is of a different nature, in that database
values are not considered atomic entities, but they are fuzzy
(membership) functions over some base set~\cite{beaubouef2012rough}. For instance, given a 
domain \emph{Age}, the fuzzy values {young} and 
{infant} can be seen as functions giving a degree of membership
 to each value in {Age}. These functions can even be
modified, e.g., the function {very young} can be considered as
modification (a translation of sorts) of {young}. Thus, when
considering an FD $X \to Y$, and two tuples $t_1,t_2$ in some table,
we usually need to decide whether $t_1[A] = t_2[A]$ (for
some $A \in X \union Y$). In the case of fuzzy databases (when
$t_1[A]$ and $t_2[A]$ can be fuzzy functions), we need to rely on
fuzzy logic to determine their degree of similarity---in other words,
we have no equality relation as a crisp, binary
relation~\cite{Raju:1988:FFD:42338.42344}. The exact concept of FD
depends on the underlying fuzzy logic being
adopted~\cite{doi:10.1080/03081079.2013.798909}. 

\section{Background}
\label{section:background}


There has been much work on FDs over incomplete information. The most
common model for incomplete information over relational databases allows 
tables with \nullname{} markers. However, when something is known about the
missing value, it might be preferable to keep
track of the various possible worlds, as advocated
by Imielinski et al.~\cite{imielinski1989incomplete, ImielinskiLipski}.
For example, we can use disjunctive
databases, in which a tuple is a disjunction of (regular) tuples,
showing the different possible values that a tuple can take. Consider the problem where Jill or Bob could be John's
superior. We can indicate these possibilities with the disjunctive
tuple $(\textrm{John}, \textrm{Jill})\;||\;
(\textrm{John}, \textrm{Bob})$ which indicates that both tuples,
$(\textrm{John}, \textrm{Jill})$ and $(\textrm{John}, \textrm{Bob})$,
are possible.
A convenient alternative is the {vague} table where 
{individual} attributes can be set-valued, indicating the possible
values of an attribute in a given tuple.  For example, we
might write $(\textrm{John},\{ \textrm{Jill},\textrm{Bob}\})$ to
indicate that John's superior might be Jill or Bob. In the rest of
this section, we  formalize these concepts.  

\subsection{Vague and Disjunctive Tables}
\label{subsection:incompletetables}

\begin{table}
\caption{Standard, vague and disjunctive tuples. Standard tuples are
instances of vague and disjunctive tuples.  \label{table:ourtuples}} \centering
\begin{tabular}{cp{10cm}}
term & definition \\\hline
standard tuple & one known value per attribute:
$(\textrm{John}, \textrm{Jill})$\\ 
vague tuple & some attributes may have a choice of values:
$(\textrm{John}, \{\textrm{Jill}, \textrm{Bob}\})$\\ 
disjunctive tuple & disjunction of standard tuples:
$(\textrm{John}, \textrm{Jill}) \; ||\;
(\textrm{John}, \textrm{Bob})$\\ 
\end{tabular}
\end{table}

Fix a schema $A_1,\ldots,A_n$ as a list of attributes with attached
domains $\dom(A_i)$ ($i=1, 2, \ldots, n$). We assume that all domains are equivalence
relations under the equality ($=$). 
 The schema corresponding to table $R$ is written
$\sch(R)$. We consider three types of tuples (see
Table~\ref{table:ourtuples}): 
\begin{itemize}
\item A standard tuple is a tuple as in the conventional relational
algebra. Formally, it is a total function from the 
schema to the Cartesian product of the domains, with $t[A_i] \in 
\dom(A_i)$. For example, $(\mathrm{John, Jill})$ might represent such
a tuple. 
\item A vague tuple is like a standard tuple where entries
are non-empty
 sets of values. Formally, a vague
tuple is a total function from the schema to 
 the Cartesian product of the finite powerset of the domains (e.g.,
 $(\mathrm{John},\{\mathrm{Jill, Bob}\})$): i.e.\ we have that 
 $t[A_i] \subset \dom(A_i)$ and  $t[A_i]$ is finite. 
For simplicity, we write singletons as elements,
 i.e.\ $\mathrm{John}$  instead of $\{\mathrm{John}\}$. We consider a
 standard tuple to be an instance of vague tuple. 
 \item A  disjunctive tuple is a
 finite disjunction of standard tuples: e.g.,  \begin{align*}(\mathrm{John,
 Jill})\;||\;(\mathrm{John, Bob}).\end{align*}
 We consider a standard tuple to be an instance of a disjunctive tuple.
 \end{itemize}
  A standard table is a set  of standard tuples, and likewise for
 vague and disjunctive tables. 
  
As usual, we can restrict tuples to a subset of attributes from the
schema. Given a table $R$, a tuple $t \in R$, and a set of
attributes $X \subseteq \sch(R)$, the tuple $t[X]$  is the \emph{projection}
of $t$ on the attributes of $X$.
 For example, given  $t=(\mathrm{John, Jill})$, we have that
 $t[\mathrm{Employee}]=(\mathrm{John})$.  
 If $t$ is a standard (vague, disjunctive) tuple, so is $t[X]$.  
 Given a table $R$, the set of all its tuples projected on $X$ is
 denoted  $\pi_X(R)= \{t[X] \mid t \in R\}$. For example, given the
 two-tuple vague table $R$ made from two attributes $A$ and $B$,
 $\{(a,\{b_1,b_2\}),(a,b_3)\}$, we have that $\pi_A(R)$ is given by the
 single tuple $(a)$ whereas $\pi_B(R)$ is given by the two tuples
 $(\{b_1,b_2\})$ and $(b_3)$. 

Valuations are defined for vague and disjunctive
 tuples. A \emph{valuation} for a   
 vague tuple is a total function that  {chooses} a single value from 
$\dom(A_i)$ for each $t[A_i]$. E.g.,  $(\mathrm{John},\mathrm{Jill})$
 is a valuation of  $(\mathrm{John},\{\mathrm{Jill,Bob}\})$.  A
 valuation for a disjunctive tuple is a function that chooses one
 disjunct per tuple. For example,   $(\mathrm{John}, \mathrm{Jill})$ is a
 valuation of
 $(\mathrm{John}, \mathrm{Jill})\;||\;(\mathrm{John}, \mathrm{Bob})$. 
 In all cases, a valuation
yields a standard tuple. 

Given two tuples $t,t'$, let $t \union t'$ be a new tuple that represents
the union of all valuations and let $t \intersection t'$ be a new
tuple representing the intersection of the valuations from $t$ and
$t'$. The union  of vague (resp.\ disjunctive) tuples is a vague
(resp.\ disjunctive) tuple. The intersection is either the empty set or
a vague (resp.\ disjunctive) tuple. E.g.,
$(\textrm{John}, \{ \textrm{Bill}, \textrm{Bob} \}) \union
(\{\textrm{John},\textrm{Julie}\}, \textrm{Bill})
=(\{\textrm{John},\textrm{Julie}\}, \{ \textrm{Bill}, \textrm{Bob} \})
$ and 
$(\textrm{John}, \{ \textrm{Bill}, \textrm{Bob} \}) \intersection
(\{\textrm{John},\textrm{Julie}\}, \textrm{Bill})
=(\textrm{John}, \textrm{Bill}) $.  

Valuations are extended to vague and disjunctive
tables in the usual way, with the specification that 
 duplicates are effectively removed in the process so that a table is
 always a set of tuples---as in the traditional relational model. Hence, each 
 valuation  on a vague or disjunctive table yields a standard table, called
 a \emph{possible world}. The set of all valuations for table $T$ is
 denoted  $\val(T)$, hence $\val(T)$ is the set of all possible
 worlds. 
 
Two standard tuples are equal if all attribute values are equal: $t=t'$ if and only if  $t[A]=t'[A]$ for
all attributes $A$. For vague and disjunctive tuples, 
two tuples are equal if they have the same possible valuations. 
As indicated by the next lemma, we can check the equality between two
vague tuples on  an attribute-by-attribute basis, but not for 
disjunctive tuples. Indeed, let $t_1=(a,b)||(a',b')$ and $t_1=(a',b)||(a,b')$, then
$t_1[A]=t_2[A]$ for both attributes $A$ whereas clearly $t_1$ and
$t_2$ are not equal.

\begin{lemma}\label{lemma:attbyatt}
Two vague or two standard tuples $t_1, t_2$ are equal on
$X$ ($t_1[X]=t_2[X]$) if and only if $t_1[A]=t_2[A]$ for all
attributes $A\in X$. This result is false for disjunctive tables. 
\end{lemma}

This distinction between vague and 
disjunctive tuples is relevant for
FDs. One of Armstrong's axioms (augmentation) says that
if $t[X]=t'[X]$ implies $t[Y]=t'[Y]$, 
then $t[XZ]=t'[XZ]$ implies $t[YZ]=t'[YZ]$.
We can generalize this result over all models where equality can be
determined attribute-by-attribute. Indeed, suppose that $t[X]=t'[X]$
implies $t[Y]=t'[Y]$.  We have that  $t[XZ]=t'[XZ]$  if and only if 
$t[X]=t'[X] \land t[Z]=t'[Z]$ which implies that
$t[Y]=t'[Y] \land t[Z]=t'[Z]$ which is equivalent to
$t[YZ]=t'[YZ]$. Unfortunately, the last step may fail in a model where
we cannot check equality attribute-by-attribute like disjunctive
tables. 

Two tables are equivalent if they have the same
 set of possible worlds. In this sense, vague tables are a subtype of
 disjunctive tables. 

\begin{lemma}
\label{lemma:type-relations}
The standard, vague and disjunctive models form a strict
hierarchy. That is,
\be
\item
For any standard table $T$, there is a vague table
  $T'$ that is equivalent to it. The reverse does not hold. 
\item
For any vague table $T$, there is a disjunctive table
  $T'$ that is equivalent to it. The reverse does not hold. 
\ee
\end{lemma}

%
%

 Tables
with \nullname{} (or unknown) markers~\cite{Levene:1999:DDI:310701.310712,
Badia15052014,levene1997additivity}  can be considered as a special
case of vague table if the domains are finite. In effect, a tuple over
schema $A_1,\ldots,A_i,\ldots,A_n$ with values
\begin{align*}(a_1,\ldots,\nullname{},\ldots,a_n)\end{align*} can be considered as a shortcut
for $ (a_1,\ldots,\dom(A_i),\ldots,a_n)$. However, it is standard to
assume infinite domains, in which case tables with \nullname{} markers
are incomparable to vague or disjunctive tables. 

Vague tables are one of the simplest
models of incomplete information and, as such, they have well-known 
expressive shortcomings. For instance, a formalism is \emph{complete}
if it can represent any finite set of possible worlds;  vague tables
are not complete. Also, a formalism is \emph{closed} (under the
relational algebra) if the result of applying (relational algebra)
operators to instances of the formalism always results in instances in
the formalism. However, the relational join of two vague tables is not 
necessarily a vague table. Hence, more powerful formalisms, like
c-tables, have been proposed~\cite{ImielinskiLipski}. However, such
models are hard to reason with; in particular, there is no
agreed-upon concept of FD for such sophisticated models. 
In fact, our
research shows that a concept of FD that meets some intuitive
requirements does not exist even for disjunctive tables, let alone more general models like c-tables. 




\subsection{Functional Dependencies}
\label{subsection:fds}

\ignore{
\begin{table}
\caption{Functional dependencies $X\to
Y$. \label{table:funcdep}}
\centering\begin{tabular}{cp{5cm}}
term & definition \\\hline
Standard FD & only defined over a standard table: $t_1[X] = t_2[X]$
then $t_1[Y] = t_2[Y]$\\[1ex]
Weak FD ($\to_W$)& There exists a world where the FD $X\to Y$ is
satisfied (~\cite{ImielinskiLipskiL})\\[1ex] 
Strong FD ($\to_S$) & The FD $X\to Y$ is satisfied in all possible
worlds (~\cite{ImielinskiLipski})\\[1ex] 
Vertical FD ($\to_V$) & Introduced by Das Sarma et
al.~\cite{DBLP:conf/amw/SarmaUW09}\\[1ex] 
PFD ($\to_{\mathrm{PFD}}$)&
given disjoint sets  $X$ and $Y$, whenever $t_1[X]$ and  $t_2[X]$ are
equal in some possible world, 
then $t_1[Y]$ and  $t_2[Y]$ must be identical (See
Definition~\ref{defin:pfd}.)  
\end{tabular}
\end{table}
}

\begin{table}
\caption{Functional dependencies $X\to
Y$. \label{table:funcdep}}
\centering\begin{tabular}{cp{5cm}}
term & definition \\\hline
Standard FD & only defined over a standard table \\[1ex]
Weak FD ($\to_W$)& Defined over incomplete tables
(~\cite{ImielinskiLipski})\\[1ex]  
Strong FD ($\to_S$) & Defined over incomplete tables
(~\cite{ImielinskiLipski})\\[1ex]  
Vertical FD ($\to_V$) & Defined over incomplete tables; introduced by
Das Sarma et al.~\cite{DBLP:conf/amw/SarmaUW09}\\[1ex] 
PFD ($\to_{\mathrm{PFD}}$)&
Defined over incomplete tables; introduced here (see
Definition~\ref{defin:pfd}.)  
\end{tabular}
\end{table}


An FD over a given schema $R$ is an expression $\fd{X}{Y}$ with $X,Y
\subseteq \sch(R)$. A standard table satisfies an FD in the
standard way: $\fd{X}{Y}$ is satisfied by $R$ if and only if  whenever
two tuples $t_1$, $t_2$ of $R$ are such 
that $t_1[X] = t_2[X]$ then $t_1[Y] =
t_2[Y]$. Of course, this definition is only directly applicable to
standard  tables (without \nullname{} markers or set-valued 
attributes\footnote{We can extend the equality relation $=$
to sets, but we reserve the equality symbol  $=$ for
elements from a domain.}). 
Table~\ref{table:funcdep} summarizes the main FD types we consider.

\paragraph{Weak and Strong FDs}

For vague or disjunctive tables, we have weak and strong satisfaction
based on possible worlds: 

\begin{defin} (Strong and weak satisfaction)
\label{def:weakstrong}
\begin{itemize}
\item 
An FD $\fd{X}{Y}$ on a vague or disjunctive table $T$
is \emph{strongly satisfied} (in symbols, $X \to_S Y$) if for all
valuations $\mu \in \val(T)$, $\mu$ satisfies $\fd{X}{Y}$
in the standard sense. 
\item 
An FD $\fd{X}{Y}$ on a vague or disjunctive table $T$ is \emph{weakly
  satisfied} (in symbols, $X \to_W Y$) if for some valuation $\mu \in
  \val(T)$, $\mu$ satisfies $\fd{X}{Y}$ in the standard sense. 
\end{itemize}
\end{defin}
\noindent
 In the conventional setting, a set of FDs is satisfied if and only if
 each FD, taken individually, is satisfied. In the context of weak
 satisfaction, this definition is unsatisfactory,  (see
 Section~\ref{section:motivation}  and Fig.~\ref{fig:examples} for
 examples), hence the next  definition.   
 
\begin{defin} (Seamless satisfaction)
A set of FDs ${\cal F}$ is \emph{seamlessly satisfied} (or holds
seamlessly) on a vague or disjunctive table $T$
if there is a valuation
$\mu \in \val(T)$ such that $\mu$ satisfies all $f \in {\cal F}$.
\end{defin}

Seamless satisfaction is not a new idea. 
 In the context of tables with unknown values, Levene and Loizou
 introduced this concept without a specific name (i.e.\ ``satisfaction
 of a set of
 FDs'')~\cite{Levene:1999:DDI:310701.310712,levene1997additivity}. 
We compare Levene and Loizou's approach to ours in
Section~\ref{section:relatedwork}.

\paragraph{Vertical FDs}

Das Sarma et al.~\cite{DBLP:conf/amw/SarmaUW09}  define 
{vertical FDs} in the context of disjunctive tables.
 To explain this concept, we need
some additional  notation.
Given a disjunctive tuple $t$  and a
 standard subtuple $\bar{a}$, let  
$t[X=\bar{a}]  = \{t' \in t \mid t'[X] = \bar{a}\}$, that is,
the selection of the disjuncts of $t$ that agree with
$\bar{a}$ on the  attributes in $X$ (if there are none, this denotes
the empty set).  Further, let
 $t[X=\bar{a}][Y]$ be the projection of the disjuncts in
 $t[X=\bar{a}]$ on $Y$: $t[X=\bar{a}][Y] = \{t'[Y] \mid t' \in
 t[X=\bar{a}]\}$. 
As an example, let $t$ be a disjunctive tuple over schema {\em
(employee, superior)} with data $(\{John,Peter\}, \{Jill,Bob\})$. 
 Then $t[\mathrm{superior}]$ results in a tuple over schema {\em
 (superior)} with data $(\{Jill,Bob\})$; 
$t[\mathrm{employee} = \mathrm{John}]$ results in a tuple over schema
{\em (employee, superior)} with data $(John, \{Jill,Bob\})$;
and $t[\mathrm{employee} = \mathrm{John}][\mathrm{superior}]$ results
in a tuple over schema {\em (superior)} with data $(\{Jill,Bob\})$.
To simplify the notation, we treat disjunctive tuples as sets of
tuples when convenient. 




With these auxiliary concepts, we can define vertical FDs as
follows. (See Fig.~\ref{fig:exampleu} for an example.) 

\begin{defin}
\label{defin:verticalfds}


The vertical FD $X \to_{\mathrm{V}} Y$ holds over a disjunctive table $R$ if
for any two tuples $t_1,t_2$\footnote{While this is not the original
definition of~\cite{DBLP:conf/amw/SarmaUW09}, it can easily be shown
to be equivalent to it, and is in a much more convenient form for our
purposes here.},
\begin{enumerate}
\item  we have that $t_1[X=\bar{a}][Y]=t_2[X=\bar{a}][Y]$ for all
$\bar{a} \in t_1[X] \intersection t_2[X]$ and  
\item $t_1[X=\bar{a}][Y-X]$ as a set of disjuncts is the Cartesian
product $t_1[X=\bar{a}][Y-X]=\prod_{A\in Y-X} t_1[X=\bar{a}][A]$, and 
\item the multivalued dependency $X \twoheadrightarrow Y-X$ holds.
\end{enumerate}
\end{defin}
%
%

\paragraph{Fuzzy Functional Dependencies}

An exhaustive review of FDs over fuzzy databases is beyond our scope:
instead, see Bosc et al.~\cite{Bosc:1998:FFD:274900.274903}. For
illustrative purposes, we nevertheless consider one notion of FD over
fuzzy sets. 

Raju and Majumdar~\cite{Raju:1988:FFD:42338.42344} define fuzzy
functional dependencies (FFD) over fuzzy sets. We can readily adapt
their definition to our context. 
Their approach starts with the introduction of a
(fuzzy) {resemblance relation}. Resemblance relations
$\mu_{\mathrm{EQ}}$ associate pairs of values to values in $[0,1]$
with the intuition that two input pairs are most alike when their
resemblance is 1. Raju and Majumdar require reflexivity
($\mu_{\mathrm{EQ}}(a,a)=1$) and symmetry
($\mu_{\mathrm{EQ}}(a,b)=\mu_{\mathrm{EQ}}(b,a)$). Though Raju and
Majumdar allow database designers to 
introduce their own, we use  the resemblance relation proposed by Raju 
and Majumdar~\cite[Example~5.1]{Raju:1988:FFD:42338.42344}: 
 given two sets $a$ and $b$, we define 
\begin{eqnarray*}\mu_{\mathrm{EQ}}(a,b)=\max(|a\cap b|/|a|, |a\cap
b|/|b|)
\end{eqnarray*}  
The value of $\mu_{\mathrm{EQ}}(a,b)$ is always in $[0,1]$. We have
that $\mu_{\mathrm{EQ}}(a,b)=0$ if and only if $a$ is disjoint from
$b$. We have that $\mu_{\mathrm{EQ}}(a,b)=1$ if and only if $a$ is
contained in $b$, or $b$ is contained in $a$.  When $a$ and $b$ are
tuple-valued, following Raju and Majumdar, we use the minimum of the
resemblance relations on each attribute: 
$\mu_{\mathrm{EQ}}((a_1,\ldots,a_n),(b_1,\ldots,b_n))=\min(\mu_{\mathrm{EQ}}^1(a_1,b_1),\ldots,\mu_{\mathrm{EQ}}^{n}(a_n,b_n))$   
where $\mu_{\mathrm{EQ}}^1, \ldots, \mu_{\mathrm{EQ}}^{n}$ are the
resemblance relations on attributes $1, \ldots, n$. 

\begin{defin}\label{def:ffd}
The Raju-Majumdar FFD 
$X\to_{\mathrm{RM}}Y$ holds if, for any two 
tuples $t,t'$, 
$\mu_{\mathrm{EQ}}(t[Y],t'[Y])\geq \mu_{\mathrm{EQ}}(t[X],t'[X])$. 
\end{defin}

\paragraph{Armstrong's axioms}
We  expect FDs to behave in certain ways: if $\fd{A}{B}$ and
$\fd{B}{C}$, we expect $\fd{A}{C}$ to hold. This is one of the
properties known as Armstrong's axioms:

\begin{defin}\label{def:Armstrong}
For standard FDs, Armstrong's axioms are the following:
 \begin{enumerate}
\item Reflexivity: If $Y\subseteq X$, then $X\to Y$. 
\item  Augmentation:  we have that  $X\to Y$ implies  $XZ\to
YZ$. 
\item Transitivity:  If $X \to Y$ and $Y \to Z$, then $X\to Z$. 
\end{enumerate}
\end{defin}

These axioms, guaranteed to hold in a standard table, characterize the
behavior of FDs.  
Many intuitive properties can be derived from Armstrong's axioms such
as the  splitting (or decomposition) rule ($X\to Y$ implies $X\to Z$
if $Z \subset Y$) and the combining (or union) rule ($X\to Y$ and
$X\to Z$ implies $X\to Y \cup Z$). 


Unfortunately, Armstrong's axioms are no longer valid when
working with vague or disjunctive tables in the weak interpretation.
In contrast, they hold for strong, 
fuzzy and vertical FDs.

\section{Desirable Properties of Functional Dependencies}
\label{section:desirable}

There are already multiple notions of FD over models of incomplete
information. We consider what constraints  should exist on such notion. 

We examine the case of a designer
faced with a table $R$, and an FD $X \to Y$ in $R$, and
ask whether any characteristics of the FD would help the designer
determine if $X \to Y$ indeed holds in $R$, or what it would take to
enforce it or, if faced with more than one FD, how the FDs as a whole 
behave. We believe that the following are likely to be helpful in such
a scenario: 
\begin{enumerate}
 \item Whether any given FD $X \to Y$ holds should depend only of $X
 \union Y$ (\emph{independence from irrelevant attributes}). Thus, in
 checking $X \to Y$ against some existing data, it should be possible
 to restrict attention and disregard the rest of the schema. 
 
To illustrate, imagine that we consider the FD
$\mathrm{employee} \to \mathrm{superior}$. Suppose that we add a new
attribute (e.g., the employee's email address): it would seem
counter-intuitive that this new, otherwise irrelevant attribute, could
invalidate the FD\@. In a distributed and dynamic setting, where
attributes can be created or removed suddenly, or where some
attributes may only be accessed remotely, it would be inconvenient to
have to check all attributes to enforce any one FD. 

 Formally, we require that the FD $X \to Y$ holds on $R$ if and
 only if it holds on $\pi_{X\cup Y}(R)$. In
 particular, adding a new attribute to a table should not violate
 existing FDs. 

This property should not be taken for granted. For example, strong FDs
 do not satisfy this condition: the table with schema (employee,
 department, manager) made of two  tuples  
\begin{align*}
 (\mathrm{Joe},\mathrm{Engineering},\{\mathrm{Gauss}, \mathrm{Newton}\}),\\ 
 (\mathrm{Jack}, \mathrm{Engineering},\{\mathrm{Gauss}, \mathrm{Newton}\})
\end{align*}
  fails to satisfy $\mathrm{department} \to_S \mathrm{manager}$ but
 its projection on $\{\mathrm{department},\mathrm{manager}\}$
 does.\footnote{The result would be different if the tables were
 multisets instead of sets~\cite{1393007}.  In such instances, the
 projection operator  might not prune duplicates and neither strong
 FDs would be satisfied. However, we wish to recover the standard
 relational algebra  when all tuples happen to be standard tuples and
 we therefore require  tables to be sets of
 tuples. 
 } 
 
  \item If an FD is satisfied in all possible worlds, then the FD should be
    satisfied  (\emph{implied by strong satisfaction}). Intuitively,
    if the FD holds in all possible worlds, then it certainly
    should hold. In particular, an FD should always hold in a table
    with a single tuple. Not all notions of FD that have been proposed
    satisfy this property: vertical
    FDs do not have this property, as pointed out by Das Sarma et
    al.~\cite{DBLP:conf/amw/SarmaUW09}.
\item If an FD is satisfied, it should be satisfied
    in some possible world (\emph{implies weak satisfaction}). A situation in
    which the FD is deemed to hold in the vague or disjunctive table,
    and no valuation is capable of satisfying the FD might be
    quite puzzling (and frustrating) for the designer.
 
 
 \item A set of FDs should satisfy Armstrong's axioms (\emph{Armstrong
 satisfaction}). In particular, we consider that transitivity has a
 strong intuitive and historical appeal. Though it might be
 possible to design database systems with the help of FDs but without
 Armstrong's axioms, we are not aware of any framework
 that supports such work. 
\item If all the FDs in a set are satisfied, then there should be a
 possible world where they are all jointly satisfied (\emph{seamless
 satisfaction}). This can be considered as a strong version of 
 the property \emph{implies weak satisfaction}. 
 
 Consider the case of a designer that is looking at a
 table $R$ and a set of FDs $F$. The designer may be 
 satisfied that each individual FD in $F$ could reasonably hold in $R$,
 perhaps even determine that there are  cases where each FD
 holds. But if there is no single possible world where all FDs hold,
 this implies that if we were to completely remove all uncertainty on
 the table, the table would be inconsistent with $F$. This
 goes against the intuition that, in satisfying each individual FD,
 the table should also satisfy the set $F$---and knowing more about the
 data in the table should not damage this process.
 \end{enumerate}

We posit that the semantics of any notion of FD should be
constrained by these basic requirements. Beyond these semantic
considerations, we look into practical issues of efficiency in  section~\ref{section:efficiency}.

\ignore{%
Table~\ref{table:properties} describes the properties satisfied by the
various types of FDs  over weak tables. 
\danielinline{What does ``weak table'' means? I guess you mean
disjunctive tables? This is an instance where things get complicated
because we have multiple theoretical models at play at the same time
in a non-trivial manner.} 
\antonioinline{The table is missing one dimension, which is the type
of incomplete table on which the property does/does not hold. I've
retired it.}
\begin{table*}
\centering 
\caption{\label{table:properties}Properties of the various types of
 functions dependencies\danielinline{This table is now nonsense. PFDs
 do not have seamless satisfaction over disjunctive tables, only over
 vague tables. Meanwhile, it would be quite a stretch to talk about
 horizontal and vertical FDs over vague tables.}\antonioinline{It
 should be possible to do so since vague tables are a subset of
 disjunctive tables, but let's not force this point.}}  
\begin{tabular}{p{5cm}ccccccc}
 & Strong & Weak & $H_1$ & $H_2$ & Vertical & PFD & RS \\\hline
independence from irrelevant attributes & No & Yes & Yes & Yes  & No &
Yes  & Yes \\ 
 implied by strong satisfaction & Yes & Yes & No & No & No & Yes   & Yes  \\
Armstrong satisfaction & Yes & No & Yes & Yes & Yes & Yes & Yes \\
Seamless satisfaction & Yes & No & Yes & No & Yes & Yes & Yes \\
Computationally practical & Yes & Yes & Yes & Yes & Yes & Yes & Yes \\
\end{tabular}
\end{table*}
}

\subsection{Conservativity, Completeness and
Soundness}

We have left out one other important property because 
it is implied by our properties.
We say that a notion of FD satisfies \emph{conservativity}
if the standard notion of FD is strictly equivalent
to the novel notion of FD when  $R$ is a standard table. More
formally, given a new FD denoted $X \to_{\mathrm{new}} Y$, then when
$R$ is a standard table, the standard FD  $X\to Y$ holds if and only
if $X \to_{\mathrm{new}} Y$ holds. Vertical FDs
satisfy \emph{conservativity}~\cite[Theorem
4.6]{sarmauncertain:2009}. The following lemma shows that any notion
of FD satisfying two of our properties automatically satisfies
conservativity.

\begin{lemma}\label{lemma:onetwothree}
The properties \emph{implied by strong satisfaction} and \emph{implies
weak satisfaction} imply \emph{conservativity}. 
\end{lemma}

\begin{proof}
Consider a standard table $R$. Write the new notion of FD as
$X \to_{\mathrm{new}} Y$ and suppose that it is both  \emph{implied
by strong satisfaction} and \emph{implies weak satisfaction}. 

 Suppose that $X\to Y$ holds, then it holds in all possible worlds
 because $R$ is a standard table, and by \emph{implied by strong
 satisfaction}, we have that $X \to_{\mathrm{new}} Y$ holds.  

Suppose that $X \to_{\mathrm{new}} Y$ holds, then by \emph{implies
weak satisfaction}, we have that $X\to Y$ must hold in at least one
possible world, but because $R$ is a standard table and there is only
one possible world, we have that $X\to Y$ must hold in the standard
sense. 
\end{proof}

Our desirable properties are not dependent on a particular model of
incomplete information. Since vague tables are a subset of 
disjunctive tables, all notions of FD that do not satisfy a given
property on vague tables also fail to satisfy it on disjunctive tables
or any generalization thereof. This is important given the
number of different  models for incomplete information,
each one with different expressive power and characteristics. 

A different list of properties could be chosen on intuitive
grounds. Our properties are slanted towards those that, we believe,
would facilitate the task of a designer in practical settings. Other
analyses may call for a different list of properties.
Moreover, these properties are not required, but we view them as desirable.

Beyond these properties, it is common to require that a new concept of
FD comes with a set of axioms which is sound and complete:  
\begin{itemize}
\item A set of axioms is sound if for
any $R$ and a set of FDs  satisfied by $R$, any FD derived from these
FDs using the axioms holds in $R$.
\item  A set of axioms is complete if  for any set of FDs, if an FD  is
true in every $R$ satisfying this set of FDs, then this FD can be
derived from the set of FDs using the axioms.
\end{itemize}
For example, Armstrong's axioms form a sound and complete set for 
standard tables. The following lemma generalizes this fact.

\begin{lemma}
\label{lemma:armstrong.conservativity}
Armstrong's axioms form a sound and complete set 
for any notion of FD satisfying both Armstrong's axioms and conservativity.
\end{lemma}

It is immediate from the preceding lemma and
Lemma~\ref{lemma:onetwothree} that our desirable properties imply that
Armstrong's axioms are a sound and complete set of axioms.

\begin{lemma}
Armstrong's axioms form a sound and complete set 
for any notion of FD satisfying our desirable properties.
\end{lemma}

\section{Comparing Existing Notions of Functional Dependencies and
Proposed Desirable Properties} 
\label{section:motivation}

In this section, we examine how previously proposed notions of FD fare with
respect to the list of desired properties. 



Strong FDs can be restrictive, e.g., given
a schema 
``$\mathrm{employee}, \mathrm{department}, \mathrm{manager}$'' subject to 
``$\mathrm{employee}\to_S \mathrm{department}$'' and
``$\mathrm{department}\to_S \mathrm{manager}$'', then a table with the
two tuples 
\begin{align*}
 (\mathrm{Joe},\mathrm{Engineering},\{\mathrm{Gauss}, \mathrm{Newton}\}),\\ 
 (\mathrm{Jack}, \mathrm{Engineering},\{\mathrm{Gauss}, \mathrm{Newton}\})
\end{align*}
fails to satisfy ``$\mathrm{department} \to_S \mathrm{manager}$''---even
though it may seem intuitive that it should. Indeed a projection on the
attributes ``$\mathrm{department}$'' and ``$\mathrm{manager}$'' (i.e.\
$(\mathrm{Engineering},\{\mathrm{Gauss}, \mathrm{Newton}\})$) would
satisfy the strong FD\@. Hence,  
adding an otherwise irrelevant attribute (``$\mathrm{employee}$'')  can
violate the FD ``$\mathrm{department}\to_S \mathrm{manager}$''. 
An FD can only hold strongly
under some constrained scenarios.

Thus we might be tempted to consider weak FDs as a viable alternative:
if there is a possible world where the FDs are satisfied, then maybe
the data is allowable. However, recall that we defined weak FDs
individually; in actual database design, sets of FDs are
considered. Thus, we have to consider situations where  a set of FDs
is satisfied, and ask if they can be satisfied seamlessly. Consider
Fig.~\ref{fig:examples}: in Fig.~\ref{fig:ex1}, there are possible
worlds where FDs $\fd{A}{B}$ and $\fd{B}{C}$ are satisfied, but both
together cannot since there is no valuation of this vague table
satisfying both FDs.  The same problem happens in Figs.~\ref{fig:ex3}
and~\ref{fig:ex4}. 
 
\begin{figure}\centering
\subfloat[$A\to B$ and $B\to C$\label{fig:ex1} ]{%
\begin{tabular}{ccc}
$A$     &   $B$       &     $C$\\\hline
a1   & \{b1,b2\}   & c1\\
a1   & \{b2,b3\}  &  c2\\
\end{tabular}
} 
\subfloat[$A\to B$ and $C\to B$\label{fig:ex3}]{%
\begin{tabular}{ccc}
$A$    &    $B$       &     $C$\\\hline
a1   & b2  &  c1\\
a1   & \{b2,b3\}  &  c2\\
a2   &      b3    &   c2\\
\end{tabular}
}
\subfloat[$A\to B$ and $C\to B$\label{fig:ex4}]{%
\begin{tabular}{ccc}
$A$    &    $B$       &     $C$\\\hline
a1   & b2  &  c1\\
a2   & b2  &  c1\\
\{a1,a2\}   & \{b2,b3\}  &  \{c2,c3\} \\
a3   &      b3    &   c2\\
a3   &      b3    &   c3\\
\end{tabular}
}

\caption{Examples of vague tables\label{fig:examples} with accompanying FDs}
\end{figure}

Thus, weak FDs, as illustrated by Fig.~\ref{fig:examples}, do not have
seamless satisfaction. Moreover, the weak interpretation does not
satisfy Armstrong's axioms. For instance,  
transitivity does not hold under the weak interpretation: we can check
in Fig.~\ref{fig:ex1} that $A\to_W B$ and $B\to_W C$ hold,
but $A\to_w C$ does not.

\paragraph{Vertical FDs}

Das Sarma et al.'s  vertical FDs satisfy seamless
satisfaction and  satisfy Armstrong's
axioms~\cite{DBLP:conf/amw/SarmaUW09}. 
We see two issues:
\begin{itemize}
\item Vertical FDs fails to provide independence from irrelevant attributes. Indeed, they  depend on multivalued dependencies. The problem with
multivalued dependencies is that they break the rule that $X \to Y$
can be checked strictly by looking at the attributes in  $X \cup Y$
(independence from irrelevant attributes). 
Consider the schema $A,B,C$ and the table made of this single
tuple $(a_1,b_1,c_1)\;||\;(a_1,b_2,c_2)$. The vertical FD  $A \to_V B$
is not satisfied because of the $C$ attribute. (If we project the
tuple on the attributes $A$ and $B$, the FD $A \to_V B$ is satisfied.)
Das Sarma et al.\ provided such
an example themselves~\cite[Example 4.5]{DBLP:conf/amw/SarmaUW09}.  
%
%
\item Vertical FDs are not implied by strong satisfaction:
that is,  an FD may hold in every single possible
world of a table, and yet it does not hold as a vertical FD\@.
Our example with the single tuple $(a_1,b_1,c_1)\;||\;(a_1,b_2,c_2)$
illustrates this exact problem.  
\end{itemize}

\paragraph{Raju-Majumdar FFDs}
There are several different notions of FDs on fuzzy databases. In
general, fuzzy FDs are appealing in part because they satisfy
Armstrong's axioms. However, some notions are too weak in the  sense
that they fail to provide seamless satisfaction. Consider the fuzzy
FDs proposed by Raju and
Majumdar~\cite{Raju:1988:FFD:42338.42344}. Let us consider 
Fig.~\ref{fig:ex3}. Both $A \to_{\mathrm{RM}} B$ and
$C \to_{\mathrm{RM}} B$ are satisfied (see
Definition~\ref{def:ffd}). Indeed, let us consider 
$A \to_{\mathrm{RM}} B$. We have that $\mu_{\mathrm{EQ}}(t[A],t'[A])$
is greater than 0 only when  $t,t'$ are the first two tuples. Yet in
this case, $\mu_{\mathrm{EQ}}(t[B],t'[B])=1$ because b2 is contained
in \{b2,b3\}. Thus the FFD  $A \to_{\mathrm{RM}} B$ holds. By a
symmetrical argument $C \to_{\mathrm{RM}} B$ holds as well. But there
is no possible world where the table from Fig.~\ref{fig:ex3} satisfies
both $A\to B$ and $C\to B$.  Fig.~\ref{fig:ex4} is another example
 of the same phenomenon but where no two tuples have exactly the same
 values for attribute $A$, or attribute $C$. In fact,
 Fig.~\ref{fig:ex4} shows that the issue persists even if one
replaces  $\mu_{\mathrm{EQ}}(a,b)=\max(|a\cap b|/|a|, |a\cap b|/|b|)$
with $\mu_{\mathrm{EQ}}(a,b)=\min(|a\cap b|/|a|, |a\cap
b|/|b|)$. Thus, the problem does not arise with just one particular
resemblance relation, and hence is not limited to
Definition~\ref{def:ffd}. While other notions of fuzzy FD may fare
better, care  is needed when choosing a particular definition of fuzzy
FDs if seamless satisfaction is desired.

We can summarize our examination of existing definitions as follows (see
Table~\ref{table:properties}):  
\begin{itemize}
\item
 Strong and vertical FDs fail to satisfy independence from irrelevant
 attributes. 
 \item 
 Vertical FDs are not implied by strong satisfaction. 
 \item Weak FDs fail to satisfy Armstrong's axioms and seamless
 satisfaction. 
 \item Though a detailed investigation of fuzzy FDs is beyond our
 scope, we have shown that some fuzzy FDs fail to provide seamless
 satisfaction~\cite{Bosc:1998:FFD:274900.274903}.
\end{itemize}

\begin{table*}
\centering 
\caption{\label{table:properties}Properties of the various types of
 functions dependencies. IIA is short for independence from irrelevant
 attributes,  SS is short for strong satisfaction,
 WS is short for weak satisfaction, Armstrong  is short
 for Armstrong satisfaction, Seamless  is short for Seamless
 satisfaction. We do not include fuzzy FDs in this table since they
 rely on different models.}
\begin{tabular}{cccccccc}
FD type & Model type &  IIA   &   SS & WS &  Armstrong  &  Seamless \\ \hline 
Strong   &   disj.\ or vague  &  No   &   Yes  & Yes &  Yes    &    Yes \\[1ex]
Weak     &   disj.\ or vague  &   Yes  &      Yes& Yes&       No  &        No\\[1ex]
Vertical & disjunctive &    No    &   No & Yes & Yes   &     Yes\\[1ex]
\multirow{2}{*}{PFD}  &    disjunctive &    Yes  &    Yes& Yes &  No   &     No\\
      &   vague       &    Yes  &  Yes  & Yes &  Yes  &      Yes\\
\end{tabular}
\end{table*}

\section{The Influence of the Table Type on FD Properties}
\label{section:results}

In view of the negative results of the previous subsections, the
question arises as to whether it is possible for a notion of FD over
tables for incomplete information to have all the desirable
properties. It turns out that the answer depends crucially on the type
of table being considered. In the next subsection, we show that it is
not possible to satisfy all properties over disjunctive
tables. However, we then show that it is possible to satisfy all the
properties over vague tables by giving a new notion of
FD, \emph{P-functional dependency (PFD)} that has all the
properties. We then show, in Section~\ref{subsection:pfdsaregood}, how
PFDs are different from other notions of FD, and provide sound and
complete axioms for them over vague tables. 

\subsection{Properties over Disjunctive Tables}
Satisfaction of all desired properties over disjunctive tables is
impossible:

\begin{theorem}\label{theorem:baddisjunctivetable}
Over disjunctive tables, there cannot be a notion of FD (denoted
$X \to_{\mathrm{magic}} Y$) satisfying the following three  properties  
\begin{itemize}
\item  \emph{implied by strong satisfaction}: if the FD $X\to Y$ holds
strongly ($X\to_S Y$) then $X \to_{\mathrm{magic}} Y$ holds, 
\item  \emph{independence from irrelevant attributes} :
$X \to_{\mathrm{magic}} Y$ holds over $R$ if and only if it holds over
$\pi_{X\union Y}(R)$; and
\item \emph{seamless  satisfaction}: given a set of FDs
$X \to_{\mathrm{magic}} Y$, each holding in some possible world, there must
be a world where all $X \to Y$ hold in the standard sense. 
\end{itemize}
\end{theorem}

\begin{proof}
Consider any notion of FD satisfying the properties 
\emph{implied by strong satisfaction}, 
 \emph{independence from irrelevant attributes}
and   \emph{seamless  satisfaction} over disjunctive tables. Let us
call this form of FD ``magic''. Observe that the  projection of the
table from Fig.~\ref{figure:badexample} to  $(A,B)$ is the {single}
disjunctive tuple $(a_1,b_1) \; || \; (a_1,b_2)$. By the
property \emph{implied by strong satisfaction}, we have that
$A\to_{\mathrm{magic}} B$ must be satisfied on this projection. By
the  \emph{independence from irrelevant attributes}, we must have that
$A\to_{\mathrm{magic}} B$ holds over the original table. 
Similarly, the projection of the table on $(C,D)$ is the {single}
tuple $(c_1,d_1) \; || \; (c_1,d_2)$ and, again by the
property \emph{implied by strong satisfaction}, we must have that
$C\to_{\mathrm{magic}} D$ is satisfied on this projection. Again, by
the  \emph{independence from irrelevant attributes}, we have that
$C\to_{\mathrm{magic}} D$ holds over the original table. 
Yet there is no possible world implied by the table where both $A\to
B$ and $C\to D$ are satisfied, thus contradicting  \emph{seamless
satisfaction}. 
\end{proof}

\begin{figure}\centering
\begin{tabular}{c}
{$R(A,B,C,D)$}\\ \hline
$(a_1,b_1,c_1,d_1)\;||\;(a_1,b_2,c_1,d_2)$\\ 
$(a_1,b_2,c_1,d_1)\;||\;(a_1,b_1,c_1,d_2)$\\
\end{tabular} 
\caption{\label{figure:badexample}Disjunctive table satisfying both
$A\to_{\mathrm{PFD}} B$ and  $C\to_{\mathrm{PFD}} D$ such that no
valuation of the table satisfying the FDs is possible.} 
\end{figure}

\subsection{P-Functional Dependencies (PFD)}
\label{subsection:pfds}

We are looking for a concept of FD that has all the desirable
properties: independence from irrelevant attributes, implied by strong
satisfaction, satisfaction of Armstrong's axioms,  seamless
satisfaction and computationally inexpensive.
Moreover, we want a concise and convenient definition to show that our desirable properties can indeed be satisfied without a contrived definition.

   Our intuition is as follows: given an FD $X\to Y$ with $X$ and $Y$
   disjoint, we want that whenever $t_1[X]$ and  
$t_2[X]$ may be equal in some possible world, then $t_1[Y]$ and 
$t_2[Y]$ must be identical.
We propose our concept of functional dependencies of possible
   worlds,  \emph{P-functional dependency}, or 
   \emph{PFD} for short.  
   (See Subsection~\ref{subsection:fds} for our notation.)

\begin{defin}
\label{defin:pfd}
A PFD $X \to_{\mathrm{PFD}} Y$ holds in $R$ whenever for any two  tuples
$t_1,t_2$, $t_1[X=\bar{a}][Y] = t_2[X=\bar{a}][Y]$ for all
$\bar{a} \in t_1[X] \intersection t_2[X]$.
\end{defin}

For example, to check that
the PDF ``$\mathrm{employee}\to_{\mathrm{PFD}} \mathrm{superior}$''
holds, we might begin by checking that the superiors given to John are
the same in every pair of tuples $t,t'$:  
\begin{align*}t[\mathrm{employee}
&= \mathrm{John}][\mathrm{superior}]\\&=t'[\mathrm{employee}
= \mathrm{John}][\mathrm{superior}].\end{align*} 
We would then repeat this check for every possible employee. Of
course, more efficient implementations using indexes are possible (see
Section~\ref{section:efficiency}).

\begin{figure}\centering
\begin{tabular}{c}
$R(A,B,C)$\\ \hline
$(a,b,c)\;||\;(a,b',c')$\\ 
$(a,b',c)\;||\;(a,b,c')$\\ 
\end{tabular}
\caption{Table $R$: Disjunctive table satisfying $A\to_{\mathrm{PDF}} C$ but not  $AB\to_{\mathrm{PDF}} CB$.}
\label{fig:newexample}
\end{figure}

As we show in Theorem~\ref{theorem:all}, no notion of FD
can satisfy all our desirable properties over disjunctive tables: PFDs
are no exception.  
PFDs fail to satisfy Armstrong's axioms and seamless satisfaction
over disjunctive tables. This is first illustrated by
Fig.~\ref{fig:newexample} where $A\to_{\mathrm{PDF}} C$ holds but not
$AB\to_{\mathrm{PDF}} CB$ which violates the augmentation axiom. 
Moreover, PFDs do not imply the existence of a seamless
valuation over disjunctive tables. Consider the  example of a
2-tuple disjunctive table in Figure~\ref{figure:badexample}.
The FDs $A\to_{\mathrm{PFD}} B$ and  $C\to_{\mathrm{PFD}} D$
hold, but no seamless valuation exists satisfying these two FDs.
\ignore{%
One aspect of this notion is noteworthy: we require
equality of the whole right-hand side, as opposed to a more relaxed
condition, like asking that $t_1[X=\bar{a}][B] = t_2[X=\bar{a}][B]$
for all $B \in Y$. The reason for the stronger requirement is that
relaxing it leads to unexpected results: consider the table $R$ in 
Fig.~\ref{fig:newexample}, and consider the PFD $A \to_{\mathrm{PFD}} \{B,C\}$:
this PFD would be satisfied in $R$ according to the weaker condition,
but there is no possible world of $R$ where it holds, thus breaking
{seamless satisfaction}. 
}

Over vague tables (and only over vague tables) PFDs satisfy the following
technical lemma that allows us to reduce any PFD $X\to_{\mathrm{PFD}}
Y$  to an equivalent set of PFDs of the form $X\to_{\mathrm{PFD}} A$
where $A$ is a singleton disjoint from $X$. 

\begin{lemma}\label{lemma:wecandecompose}
Over vague tables, $X\to_{\mathrm{PFD}} Y$ holds if and only if
$X\to_{\mathrm{PFD}} A$ for all attributes $A\in Y-X$.
Moreover, we can check that such  PFD $X\to_{\mathrm{PFD}} Y$ where
$Y$ is disjoint from $X$ holds by checking that whenever
$t_1[X] \intersection t_2[X]$ then $t_1[Y]=t_2[Y]$. 
\end{lemma}

The proof of this Lemma is left for the Appendix.

\paragraph{PFDs vs. 
Vertical FDs}

There are  differences between vertical FDs and PFDs. Consider
the Table~$U$ in Fig.~\ref{fig:exampleu}. In this table, the
vertical FD $\mathrm{SSN} \to_V \mathrm{Name}$ does not hold, but the
PFD $\mathrm{SSN} \to_{\mathrm{PFD}} \mathrm{Name}$ does.

\begin{figure}\centering
\begin{tabular}{cc}
{ID} & {$U$(SSN, Name)} \\ \hline
$r_1$ & (1,Tom) $||$ (1, Thomas) \\ 
$r_2$ & (1,Tom) $||$ (1,Thomas) $||$ (2,Tom)$||$(2,Thomas) \\ 
\end{tabular}
\caption{Table $U$: the
vertical FD $\mathrm{SSN} \to_V \mathrm{Name}$ does not hold while the
PFD $\mathrm{SSN} \to_{\mathrm{PFD}} \mathrm{Name}$ does.} 
\label{fig:exampleu}
\end{figure}

If one compares the definitions of PFDs and vertical FDs, it is
immediate that vertical FDs are strictly stronger than PFDs. 

\begin{corollary}
If $X \to_V Y$ then $X \to_{\mathrm{PFD}} Y$.
\end{corollary}

\paragraph{PFDs vs. Fuzzy FDs}
Recall that a fuzzy Raju-Majumdar FD holds whenever 
\[\mu_{\mathrm{EQ}}(t[Y],t'[Y])\geq \mu_{\mathrm{EQ}}(t[X],t'[X])\] for
all pairs of tuples $t',t'$. We can express PFDs is similar terms.
For simplicity, suppose that $X$ and $Y$ are disjoint
($X \intersection Y=\emptyset$). Consider the PFD
$X \to_{\mathrm{PFD}} Y$. The FD applies to tuples $t$ and $t'$ if and
only if there exists $\bar{a} \in t[X] \intersection t'[X]$; and in such
cases $t[Y]=t'[Y]$. In turn, over vague tables,  there exists such a
$\bar{a}$ if and only if $\mu_{\mathrm{EQ}}(t[X],t'[X])>0$.
 Thus, an FD holds if whenever $\mu_{\mathrm{EQ}}(t[X],t'[X])>0$ then
$t[Y]=t'[Y]$.  This is strictly stronger than
$\mu_{\mathrm{EQ}}(t[Y],t'[Y])\geq \mu_{\mathrm{EQ}}(t[X],t'[X])$. In
other words, PFDs are stronger than Raju-Majumdar fuzzy FDs. 


\subsection{PFDs and the Desirable Properties}
\label{subsection:pfdsaregood}

In this section, we show that PFDs, in addition to being concisely
defined, have all the properties listed in
Section~\ref{section:desirable} over vague tables, establishing Theorem~\ref{theorem:all}.

\begin{theorem}
\label{theorem:all}
PFDs satisfy the following properties over vague tables:
\begin{enumerate}
\item independence from irrelevant attributes,
\item implied by strong satisfaction,
\item implies weak satisfaction,
\item Armstrong satisfaction,
\item  seamless satisfaction,
\end{enumerate}
\end{theorem}

 We check each property separately.

\paragraph{Independence from Irrelevant Attributes} 
A PFD is independent 
from irrelevant attributes; by examining our definition, we can check
$X \to_{PFD} Y$ either by looking on $R$ or on $\pi_{X\cup Y}(R)$
alone. That is, $X \to_{PFD} Y$ holds in $R$ if and only if it holds in
$\pi_{X\cup Y}(R)$. 

\paragraph{Between Strong and Weak}
We also have that the new definition is ``in between'' strong and weak
satisfaction:
$X \to_S Y \Rightarrow X \to_{\mathrm{PFD}} Y $ and
$X \to_{\mathrm{PFD}} Y \Rightarrow X \to_W Y$. 

\begin{lemma} (Implied by strong satisfaction)
$X \to_S Y \Rightarrow X \to_{\mathrm{PFD}} Y $.
\label{lemma:strongweak}
\end{lemma}

\begin{lemma} (Implies weak satisfaction)
$X \to_{\mathrm{PFD}} Y \Rightarrow X \to_W Y$.
\label{lemma:pfdweak}
\end{lemma}

\paragraph{Armstrong satisfaction}
While augmentation is not satisfied over disjunctive tables, PFDs
satisfy all of Armstrong's axioms over vague tables.

\begin{lemma} (Armstrong satisfaction)
PFDs obey Armstrong's axioms over vague tables.
\label{lemma:armstrong}
\end{lemma}

The proofs of the previous lemmas are in the Appendix. Together, these
lemmas provide an additional, important result:
Armstrong's axioms are a sound and complete set for PFDs over vague
tables. 

\begin{lemma}
\label{lemma:armstrongsoundcomplete}
Armstrong's axioms are a sound and complete set for PFDs over vague
tables. 
\end{lemma}
\begin{proof}
Lemma~\ref{lemma:strongweak} and Lemma~\ref{lemma:pfdweak}
together imply, thanks to Lemma~\ref{lemma:onetwothree} that PFDs are
conservative. This, together with the previous Lemma and
Lemma~\ref{lemma:armstrong.conservativity}, establishes the result.
\end{proof}

\paragraph{Seamless  Satisfaction on Vague Tables} 
 Seamless satisfaction only holds on vague tables.

\begin{lemma} (Seamless satisfaction)
If each PFD in a set of PFDs is satisfied over vague table $T$, then
  the set is seamlessly satisfied over $T$.
\end{lemma}

\begin{proof}
The proof follows through a construction, see
Algorithm~\ref{algo:val}. The algorithm terminates since it iterates
once over each attribute, and then it loops once over each tuple. The
algorithm only modifies the input table in the loop at
lines~\ref{line:beg}--\ref{line:end}. We can check by inspection that
if the PFDs held before this loop, then they are still satisfied
afterward. When the algorithm terminates, what remains is a
valuation. This concludes the proof. 
\end{proof}

\begin{example}
As a simple example of application of Algorithm~\ref{algo:val}, consider the
following vague table $R(A,B,C)$, $\{(a_1,\{b_1,b_2\},c1),
(a_1,\{b_1,b_2\},c2), (a_2,\{b_1,b_2\},c2)\}$, subjected to two PFDs:
$A\to_{\mathrm{PFD}}B$ and  
$C\to_{\mathrm{PFD}}B$. The algorithm visits all attributes in sequence.
It may first consider the attribute ($A$). It would then seek all
FDs that determine $A$: in this instance, there are none. It would then
visit each tuple $t$ in sequence and check if $t[A]$ contains more
than one value. In this instance, it would find none. The algorithm
could then check attribute $C$ and do nothing again. It may then
move to  attribute ($B$). It would seek all FDs that determine $B$.
In this instance, there are two: $A\to_{\mathrm{PFD}}B$ and 
$C\to_{\mathrm{PFD}}B$.
 Consequently, it would set $\mathcal{X} = \{\{A\},\{C\}\}$ 
 to keep track of the determining sets of attributes (the LHS of
 FDs). Then it visits the 
tuples $t$ in sequence, looking for a tuple such that $t[B]$ contains more
than one value. It would encounter $t=(a_1,\{b_1,b_2\},c1)$. 
It would then pick any one value for attribute $B$, e.g., $b_1$ (the choice
can be randomized). It would then store $t$ in the initially empty
set $s$. It would then visit all tuples $t'$ in $R$, and check if there is
a possible world where $t'[X]$ might agree with $t''[X]$ for some
$t''\in s$ and some $X\in \mathcal{X}$; when that is true, it would
add the tuple to $s$. In this instance, it would add both remaining
tuples in $s$, so that all tuples from $R$ end up in $s$. It would
then modify all tuples $t'\in s$ so that $t'[B]=b_1$. 
The algorithm would terminate with the standard table 
$(a_1,b_1,c1), (a_1,b_1,c2), 
(a_2,b_1,c2)$.
\end{example}

\begin{algorithm}
\caption{Valuation algorithm for vague
tables.\label{algo:val}}\label{alg:valuation} 
\begin{algorithmic}[1]
\STATE \textbf{input}: A set of PFDs $\mathcal{F}$, of the form $X\to
Y$. Assume without loss of generality that $Y\cap X = \emptyset $ and 
$Y$ is  a singleton (see
Lemma~\ref{lemma:wecandecompose}). 
\STATE \textbf{input}: A vague table $R$ satisfying the  PFDs
\STATE \textbf{output}: A valuation of $R$ that satisfies all PFDs in
$\mathcal{F}$ 
\FOR{each attribute $A$ in the schema of $R$}
\STATE Given $A$, identify all the FDs in  $\mathcal{F}$ of the form
$X \to A$ for some set of attributes $X$. Call $\mathcal{X}$ the set
of all sets of attributes $X$. 
\FOR{each tuple $t$ in $R$}
\IF{$t[A]$ contains more than one attribute value}
\STATE Select any one attribute value $a$ in $t[A]$.
\STATE Let $s$ be a set containing initially $t$
\FOR{each tuple $t'$ in $R$}
\IF{there is a possible world where $t'[X]$ is equal to  $t''[X]$ for
some $X$ in $\mathcal{X}$ and $t''$ in $s$} 
\STATE add $t'$ to $s$
\ENDIF
\ENDFOR
\FOR{each tuple $t'$ in $s$}\label{line:beg}
\STATE set $t'[A]$ to the singleton containing $a$
\ENDFOR\label{line:end}
\ENDIF
\ENDFOR 
\ENDFOR
\RETURN $R$
\end{algorithmic}
\end{algorithm}

\section{Efficiency Issues}
\label{section:efficiency}
In addition to the above constraints on the semantics of
FDs, there is another aspect that has considerable practical impact:
computational efficiency.
Checking that an FD holds should be reasonably inexpensive
 computationally (\emph{computationally practical}). At a minimum, the
 problem should not be  NP-hard; ideally, it should be practical
 to check the FD over  large datasets. This means that the designer
 knows it is possible to enforce the FDs in a realistic database. For our
 purposes, we identify ``practical'' with sub-quadratic running times,
 e.g., linearithmic ($\BigO{n \log n}$ where $n$ is the number of tuples).  
It turns out that there are some difficult problems in regard to this
issue, as we show next.

\subsection{Checking Weak Seamless Satisfaction in Vague Tables is
NP-complete} 
\label{subsection:satisfaction}

We have seen that weak FDs fail to satisfy  Armstrong's axioms
(indeed, weak FDs are not transitive).  A closely related problem is
that weak FDs do not enforce seamless 
 satisfaction. However, we could separately require that sets of FDs
 must be jointly satisfied, and then transitivity would no longer be
 an issue.  Unfortunately, this might be computationally unfeasible,
 as the next  theorem shows.

\begin{proposition}\label{proposition:npcomplete}
Consider any set of weak FDs ${\cal F}$ over some vague table.
Determining whether such a set ${\cal F}$ holds seamlessly is
NP-complete. 
\end{proposition}
\begin{proof}
 We prove that the problem is NP-complete by reducing the 
3-Dimensional Matching (3DM) problem to it. The 3DM
problem can be 
described as follows.  We have disjoint sets $X$, $Y$, and $Z$, each
having cardinality $n$. Moreover, we have a set of triples $T$ in $X
\times Y \times Z$. We want to determine whether there is a subset of
$T$  of triples such that each element of $X \cup Y \cup Z$ is
included in exactly one tuple.
 This problem is an NP-complete problem  (\cite{Karp:1990:OAO:100216.100262}). 

Start with any 3DM instance. We want to reduce it to our problem.
Given table $T$, we create a table $T'$ with 4~attributes, one for
each of $X$, $Y$, $Z$, and then a fourth one which contains a tuple identifier
for each tuple in $T$. We abuse the notation and call these attributes
$X, Y, Z, T$. We then populate $T'$ from $T$ as follows:
given a value $x$ of $X$ (resp.\ $Y$ and $Z$) insert the tuple $(x,
{\cal Y}, {\cal Z}, t)$ where ${\cal Y}$ (resp.\ ${\cal Z}$) is the set of all
values in $Y$ (resp.\ $Z$) and $t$ is the list of tuple identifiers for tuples
containing $x$ in $T$. We end up with a table with $3n$~tuples. 
Given the set of FDs $\{\fd{X}{T}, \fd{Y}{T}, \fd{Z}{T}\}$,
determining whether there is a valuation of this table that respects
the set is equivalent to the 3DM problem.   Indeed, the valuation must
map each value $x \in X$ (resp.\ $y \in Y$, $x 
\in Z$) to one, and only one, tuple identifier because of the FD
$\fd{X}{T}$ (resp.\ $\fd{Y}{T}, \fd{Z}{T}$). By construction, each
tuple identifier corresponds to one and only one value of attribute
$X$ (resp.\ $Y$, $Z$), so the relationship between the $n$~values of
$X$  (resp.\ $Y$, $Z$) and the $n$~tuple identifiers is bijective. 
Hence, such a valuation must contain $n$~distinct tuples, and the set
of these distinct tuples corresponds to a subset of the original table
$T$. 

This concludes the proof.
\end{proof}

To illustrate the proof, let us consider an actual example.
We have that $X=\{a,b,c\}$, $Y=\{1,2,3\}$ and $Z=\{A,B,C\}$.
Let $T$ be the table in Fig.~\ref{fig:npcomplete1}.
First, we add tuple identifiers: see Fig.~\ref{fig:npcomplete2}.
And then we transform the result into vague table $T'$ with the procedure
given in the proof (see Fig.~\ref{fig:npcomplete3}). Finding a
standard table where all the FDs $X \to T$, $Y \to T$, $Z \to T$ hold
(as in Fig.~\ref{fig:npcomplete4}) is equivalent to finding a 
solution to the 3DM problem on the original instance.

\begin{figure}\centering
\subfloat[Set of triples $T$ in $X
\times Y \times Z$\label{fig:npcomplete1} ]{%
\begin{tabular}{ccc}
X & Y  &Z \\ \hline
a & 2 & B \\ 
b & 1 & A \\ 
c & 3 & C \\ 
a & 1 & B\\ 
b & 3 & B\\ 
\end{tabular}
} 
\subfloat[Table $T$ with tuples identifiers\label{fig:npcomplete2}]{%
\begin{tabular}{cccc}
X & Y  &Z & T \\ \hline
a & 2 & B& $t_1$ \\ 
b & 1 & A & $t_2$\\ 
c & 3 & C & $t_3$\\ 
a & 1 & B & $t_4$\\ 
b & 3 & B & $t_5$\\ 
\end{tabular}
} 
\subfloat[Vague table $T'$\label{fig:npcomplete3}]{%
\begin{tabular}{cccc}
X & Y  &Z & T\\ \hline
 a & ${\cal Y}$ & ${\cal Z}$ & $\{t_1,t_4\}$\\ 
 b & ${\cal Y}$& ${\cal Z}$ & $\{t_2,t_5\}$\\ 
 c & ${\cal Y}$& ${\cal Z}$ & $t_3$\\
${\cal X}$ & 1& ${\cal Z}$ & $\{t_4,t_2\}$\\ 
${\cal X}$ & 2& ${\cal Z}$ & $t_1$\\ 
${\cal X}$ & 3& ${\cal Z}$ & $\{t_3,t_5\}$\\ 
${\cal X}$ & ${\cal Y}$& A & $t_2$\\ 
${\cal X}$ & ${\cal Y}$& B & $\{t_4,t_1,t_5\}$\\ 
${\cal X}$ & ${\cal Y}$& C & $t_3$\\ 
\end{tabular}
}
\subfloat[Valuation of $T'$ satisfying $X\to T$, $Y\to T$, and $Z\to T$ \label{fig:npcomplete4}]{%
\begin{tabular}{cccc}
X & Y  &Z & T\\ \hline
 a & 2& B & $t_1$\\ 
b & 1& A &  $t_2$\\
c& 3 & C& $t_3$\\
b & 1& A &  $t_2$\\
 a & 2& B & $t_1$\\ 
c& 3 & C& $t_3$\\
b & 1& A &  $t_2$\\
 a & 2& B & $t_1$\\ 
c& 3 & C& $t_3$\\
\end{tabular}
}
\caption{Illustration of the proof of Proposition~\ref{proposition:npcomplete}\label{fig:npcomplete}}
\end{figure}

Since vague tables are a subset of disjunctive databases, we
immediately have the following corollary.

\begin{corollary}
Seamless satisfaction of standard FDs in disjunctive databases is
NP-complete. 
\end{corollary}

Hence,  weak FDs are probably not appropriate if we desire seamless
satisfaction. However, if we are willing to restrict the set of FDs to
those with
``monodependence''~\cite{Levene:1999:DDI:310701.310712,levene1997additivity},
then seamless satisfaction is ensured. But as we explain in
Section~\ref{section:relatedwork}, monodependency is quite
restrictive. 

Even if we were willing to solve potentially difficult
computational problems to enforce seamless satisfaction with weak FDs,
an approach that enforces seamless satisfaction by processing all FDs
at once would fail to meet our \emph{independence from irrelevant
attributes} property.   

\subsection{PFDs Are Computationally Practical}
The above leaves an important question open: is the new concept of PFD
computationally practical, at least in some model?
We can show that deciding whether PFDs hold in a vague table can
be determined efficiently. This is achieved by indexing the 
table based on the attributes on the left-hand-side of the PFD. 

Recall that a PFD $X\to_{\mathrm{PFD}} Y$ holds if and only if 
 for any two tuples $t,t'$, and any possible values $\bar{a}$ of
$X$, we have that $t[X=\bar{a}][Y]= t'[X=\bar{a}][Y]$. Thus if we
 build a map from the possible values $\bar{a}$ of $t[X]$ for all
 tuples $t$, and the corresponding values $t[X=\bar{a}][Y]$, then we
 can readily check if any new tuple satisfy the PFD\@.  

Let us formalize the process. Given an initially empty vague (or
disjunctive) table $R$, let $M_{X,Y}(R)$ be an initially empty map
(e.g., a hash table) from  standard tuples over $X$ to pairs of
values: one  vague (or disjunctive) tuples over $Y$ and one
non-negative integer counter. It suffices to be able to check that
insertions and removals do not violate the PFD---a tuple update can
always be implemented as a removal followed by an insertion.  
\begin{itemize}
\item When a new tuple $t$ is inserted in $R$, we check if $\bar {a}$
is present in $M_{X,Y}(R)$ for each $\bar {a}\in t[X]$. If it is not,
then we store the mapping $\bar {a}\to (t[Y],1)$ in $M_{X,Y}(R)$ where
the number 1 indicates that exactly one tuple supports the
relation. If  $\bar {a}$ is present  $M_{X,Y}(R)$, we check that the
stored value agrees with $t[Y]$, if it does not, then  $R$ does not
satisfy the PFD (and so we must reject tuple $t$), if it does, simply
increment the counter by one.  
\item When a tuple $t$ is removed from $R$, for each $\bar {a}\in
t[X]$, we  get the corresponding mapping in $M_{X,Y}(R)$ and decrement
the counter. If the counter falls to zero, we remove the corresponding 
mapping from $M_{X,Y}(R)$. 
\end{itemize}
If we use a hash table as the underlying data structure for
$M_{X,Y}(R)$, then the insertions and removals  require only  expected
linear time with respect to the number of distinct valuations of
$t[X]$. 

Thus we have shown that PFDs can be enforced efficiently on 
 vague tables. Further physical designs issues are outside our scope
 and left as future  work.

\section{Comparison with Prior Work}

A notion of seamless satisfaction was introduced
by Levene and
Loizou~\cite{Levene:1999:DDI:310701.310712,levene1997additivity}
without a name (it was called ``satisfaction of a set of 
FDs'') for tables with \nullname{} markers. Levene and Loizou~\cite{levene1997additivity} also
introduced the related notion 
of \emph{additivity}. A set $F$ has
additivity if satisfaction of each FD guarantees satisfaction of the
\emph{reduced cover} of $F$ as a set ($G$ is a cover of $F$ if and only if the
closure of $G$ is the same as the closure of $F$; $G$ is reduced if
each FD is in a minimal form). Levene and Loizou prove that, in the
context of tables with \nullname{} markers, a set of FDs $F$ is
additive if and only if it is \emph{monodependent}. Monodependency
means that, in the closure of $F$, each attribute is determined only
by one FD, and
there are no non-trivial cycles in the set (a cycle $XB \to A$ and
$YA\to B$ is trivial if either $Y\to B$ or $(X\cap Y)A \to
B$). We can  check  in polynomial time whether a set of FDs
is monodependent by constructing a canonical cover of the
set~\cite{levene1997additivity}.  However, monodependency is defined on the
\emph{closure} of the set of FDs; as a consequence, sets like $\{A\to
B,\ B\to C\}$ are not  monodependent. Thus, while quite powerful, the
notion of monodependence is also restrictive.  In contrast 
to this approach, we work with the set of FDs as given, without
relying on a notion of closure. We also use a more general
framework, 
focusing on vague and disjunctive tables. In this sense, we are
following a suggestion of Levene and
Loizou~\cite[p. 13]{levene1997additivity}, who write: ``It would be 
an interesting research topic to extend the results presented herein
to  or-sets, i.e.\ allowing, instead of any  occurrence of unk, a finite
set of possible values, one of which is the true value.''
We have shown that checking
seamless satisfaction in vague tables (hence, in disjunctive tables) is
NP-complete. Also, we have shown that a set of FDs cannot be seamlessly
satisfied over disjunctive tables if such set must also be subject to
the properties  \emph{implied by 
strong satisfaction} and \emph{independence from irrelevant attributes}.
Thus, our work can be seen as an initial step
in developing Levene and Loizou's research suggestion. 
 
In other work on FDs with \nullname{} markers~\cite{Badia15052014},
two interpretations are proposed: literal and super-reflexive FDs. In
this work, it is required that FDs be 
realizable, which is equivalent to \emph{implies weak satisfaction}, and
to be strongly realizable, which is equivalent to \emph{seamless
 satisfaction}. It is also required that \emph{Armstrong
 satisfaction} holds and that checking FDs be \emph{computationally practical}.
 In a manner that is reminiscent (but distinct) from 
 Levene and Loizou~\cite{Levene:1999:DDI:310701.310712,levene1997additivity},
 sets of super-reflexive FDs must be cycle-free and  have at most
 one FD determining any one attribute, to support \emph{seamless
 satisfaction}. However, unlike Levene and Loizou's monodependency
 result, the conditions applied on the set of FDs as given, and not to 
 their closure:  a chain of super-reflexive FDs such as
 $A\to B$ and $B\to C$ supports  \emph{seamless
 satisfaction}.

Another area of research focuses on different models to represent
 incomplete information. There is a trade-off
 between the  expressivity and the complexity of such models: simple
 models are  intuitively easier to understand, 
 but limited in what they can express, while complex models are 
 expressive but can be  difficult to reason with. Since we are mostly
 interested in practical design applications, we have 
 focused on two models which are at the bottom and middle of the hierarchy
 of Das Sarma et al.~\cite{sarmauncertain:2009}: our vague tables are
 called {attribute-or} tables there, and are the simplest and
 less expressive  model studied, while disjunctive tables are in the
 middle, more  powerful than vague tables but less powerful than
 models like  c-tables~\cite{ImielinskiLipski}. C-tables allow
 variables in tuples 
 and arbitrary formulas over such variables, so as to express
 constraints on the values that such variables can take. 

 In a different context (probabilistic XML), Amarilli found that
determining whether a document is a possible world given a
probabilistic document is
NP-hard~\cite{amarillipossibility}. Similarly, we determined 
NP-hardness for checking seamless satisfaction.

\section{Conclusion and Further Research}
\label{section:conclusion}

We have proposed a set of basic properties that we consider desirable
for any definition of FDs to fulfill from a conceptual design point of
view. Our properties are independent of the underlying model of
incomplete information  (see Section~\ref{section:desirable}). 

We have argued that such properties would make the conceptual design
of the database easier in many  real-world scenarios. To our
knowledge, no prior work has considered a comparable set of
desiderata for the semantics of FDs over incomplete information.  

We have then examined several proposals in the literature, including
weak and strong satisfaction, the vertical FDs of Das Sarma
et al.~\cite{DBLP:conf/amw/SarmaUW09} and the fuzzy dependencies
of Raju and Majumdar~\cite{Raju:1988:FFD:42338.42344}, 
and shown that they do not satisfy one or more of the desired
properties. In particular, we have proven that checking whether a set
of FDs is seamlessly satisfied under weak satisfaction is NP-complete
(see Proposition~\ref{proposition:npcomplete}), 
thus not meeting the computational efficiency property.
 Finally, we introduced {P-Functional FDs (PFDs)} as a simple
 notion of FD, and we have proven that PFDs have all the desired
 properties on vague tables. Therefore we have shown that while no
 notion of FD can satisfy all our desirable properties over
 disjunctive tables (Theorem~\ref{theorem:all}), it is
 possible to satisfy all of them over vague tables. 
Moreover, we have shown that Armstrong's axioms form a complete
and sound axiomatization for PFDs over vague tables.

While we have argued intuitively for the desirability of the proposed
properties, other sets of properties could be investigated. From a
different perspective, it would be interesting to determine maximal
and practical useful sets of desirable properties that \emph{can} be
satisfied in powerful models of incomplete information, like
disjunctive databases and c-tables. 

\section*{Acknowledgments} 

This work was supported by the Natural Sciences and Engineering
Research Council of Canada [26143].

\bibliographystyle{fundam}

\bibliography{new.vague}

\section*{Appendix}

\paragraph{Proof of Lemma~\ref{lemma:wecandecompose}} Over vague tables,
$X\to_{\mathrm{PFD}} Y$ holds if and only if $X\to_{\mathrm{PFD}} A$
for all attributes $A\in Y-X$. 
\begin{proof}
By definition, if the $X\to_{\mathrm{PFD}} Y$ holds, then 
 $t_1[X=\bar{a}][Y] = t_2[X=\bar{a}][Y]$ for all
$\bar{a} \in t_1[X] \intersection t_2[X]$, and therefore $t_1[X=\bar{a}][A]
 = t_2[X=\bar{a}][A]$ for all  $A\in Y-X$. Let us consider the reverse
 implication. Suppose that $X\to_{\mathrm{PFD}} A$ for all attributes
 $A\in Y-X$. We have that $t_1[X=\bar{a}][A] = t_2[X=\bar{a}][A]$ for
 all $\bar{a} \in t_1[X] \intersection t_2[X]$ not only for all $A\in
 Y-X$, but clearly for all $A\in Y$. Indeed, if $A\in
 X \intersection Y$, then both $t_1[X=\bar{a}][A]$ and
 $t_1[X=\bar{a}][A]$  must agree with $\bar{a}$ projected on $A$. By
 Lemma~\ref{lemma:attbyatt}, the rest of the result follows. 
\end{proof}

\paragraph{Proof of Lemma~\ref{lemma:strongweak}} $X \to_S Y \Rightarrow
X \to_{\mathrm{PFD}} Y $. 
\begin{proof}
Assume $X \to_S Y$ but
  $X \not \to_{\mathrm{PFD}} Y$. Then there are tuples $t_1$, $t_2$ with
  $t_1[X= \bar{a}][Y] \neq t_2[X= \bar{a}][Y]$ for some $\bar{a} \in
  t_1[X= \bar{a}] \intersection t_2[X= \bar{a}]$. Assume without loss of
  generality that there is 
  $\bar{b} \in t_1[X= \bar{a}][Y]$ and $\bar{b} \notin t_2[X=
    \bar{a}][Y]$. Then there is a possible world which assigns, to
  $t_1$, $\bar{a}\bar{b}$ for $t_1[XY]$ and $\bar{a}\bar{c}$ for
  $t_2[XY]$, $\bar{c} \neq \bar{b}$ (since $\bar{b} \notin t_2[X=
    \bar{a}][Y]$). In this world, $X \to Y$ does not hold (in standard
  form); hence $X \to_S Y$ fails. Contradiction.
\end{proof}

\paragraph{Proof of Lemma~\ref{lemma:pfdweak}} $X \to_{\mathrm{PFD}}
Y \Rightarrow X \to_W Y$. 
\begin{proof}
Assume $X \to_{\mathrm{PFD}} Y$,
  and $X \not \to_W Y$. Then, there is no possible world where $X \to
  Y$ holds in a standard manner. This can only happen if there are tuples
  $t_1$, $t_2$ with $t_1[X] \intersection t_2[X] \neq \emptyset$ and
  $t_1[Y] \intersection t_2[Y] = \emptyset$. But then there is some $\bar{a}
  \in t_1[X] \intersection t_2[X]$, and for that $\bar{a}$, we have that
  $t_1[X=\bar{a}][Y] \neq t_2[X=\bar{a}][Y]$, hence contradicting the
  assumption that $X \to_{\mathrm{PFD}} Y$. 
\end{proof}

\paragraph{Proof of Lemma~\ref{lemma:armstrong}} PFDs obey Armstrong's axioms
over vague tables. 
\begin{proof}
(Reflexivity) To prove that $X \to_{\mathrm{PFD}} X$, we unravel the definition:
  this means that for any two tuples $t_1$, $t_2$, $t_1[X=\bar{a}][X]
  = t_2[X=\bar{a}][X]$. But for any $t$,  $t[X=\bar{a}][X] = \{t[X]
  \mid t[X=\bar{a}]\} = \{t[X] \mid t[X] = \bar{a}\} = \{\bar{a}\}$. 

(Augmentation) Assume $X \to_{\mathrm{PFD}} Y$. Consider
any two tuples  $t_1$ and $t_2$ that agree on $XZ$, that is, there are 
values  $\bar{a}$ for $X$ and $\bar{b}$ for $Z$ such that
$\bar{a}\bar{b} \in t_1[XZ] \intersection t_2[XZ]$. Then we have that
$\bar{a} \in t_1[X] \intersection t_2[X]$. By the assumption that
$X \to_{\mathrm{PFD}} Y$, this implies that
$t_1[X=\bar{a}][Y]=t_2[X=\bar{a}][Y]$. Therefore, we have
$t_1[XZ=\bar{a}\bar{b}][Y]=t_2[XZ=\bar{a}\bar{b}][Y]$. 
Because the projection of $t_1[XZ=\bar{a}\bar{b}]$ and
$t_2[XZ=\bar{a}\bar{b}]$ on $Z$ must agree with $\bar{b}$, we finally
have $t_1[XZ=\bar{a}\bar{b}][YZ]=t_2[XZ=\bar{a}\bar{b}][YZ]$
establishing that $XZ \to_{\mathrm{PFD}} YZ$. 

(Transitivity) Suppose that $X \to_{\mathrm{PFD}} Y$ and
  $Y \to_{\mathrm{PFD}} Z$ hold over some vague table. If 
  two tuples $t_1, t_2$ agree on $X$ with value $\bar{x}$, then we
  want to show that   $t_1[X=\bar{x}][Z]$ is identical to
  $t_2[X=\bar{x}][Z]$. 
 By $X\to_{\mathrm{PFD}} Y$, we have that $t_1[X=\bar{x}][Y]$ is identical to
 $t_2[X=\bar{x}][Y]$. Both of these are equal to the same set of tuples
 over $Y$. Consider that $Y\to_{\mathrm{PFD}} Z$, and pick a tuple $\bar{y}$ in
 $t_1[X=\bar{x}][Y]=t_2[X=\bar{x}][Y]$. We have that  $t_1[Y=\bar{y}][Z]$ is
 identical as a set to  $t_2[Y=\bar{y}][Z]$. We have that
 $t_1[X=\bar{x}][Z]$ and $t_2[X=\bar{x}][Z]$ are the unions of
 $t_1[Y=\bar{y}][Z]$ and $t_2[Y=\bar{y}][Z]$ respectively over the tuples
 $\bar{y}$ in $t_1[X=\bar{x}][Y]=t_2[X=\bar{x}][Y]$; they are therefore
 identical. 
\end{proof}

\end{document}